\documentclass{article}
\usepackage{amsfonts}
\usepackage{graphicx}
\usepackage{amsmath}
\usepackage{amsthm}
\usepackage{latexsym}
\usepackage{color}
\usepackage{times}
\usepackage{fancyhdr}
\usepackage{amssymb,amsmath}
\usepackage{lastpage}
\usepackage{epsfig}
\usepackage{algorithm2e}
\usepackage[compress]{cite}

\bibliography{references.bib}

\makeatletter

\newcommand{\Rmnum}[1]{\expandafter\@slowromancap\romannumeral #1@}
\makeatother

\renewenvironment{proof}{\par\noindent{\itshape Proof:}}{\hspace*{\fill}$\Box$\par}

\usepackage{NetCod09}  

\begin{document}
\topmargin = 0mm

\itwtitle{Convolutional Network Coding Based on \\Matrix Power Series
Representation}

\itwauthor{Wangmei Guo, Ning Cai}
 {The State Key Lab. of ISN,
 \\Xidian University, China
 \\Email:wangmeiguo, caining@mail.xidian.edu.cn}

 \itwthirdauthor{Qifu Tyler Sun}
 {Institute of Network Coding,
 \\ The Chinese University of Hong Kong,
 \\ Hong Kong SAR, China
  \\Email:qfsun@inc.cuhk.edu.hk}

\itwmaketitle



\begin{itwabstract}
In this paper, convolutional network coding is formulated by means of matrix power series representation of the local encoding kernel (LEK) matrices and global encoding kernel (GEK) matrices to establish its theoretical fundamentals for practical implementations. From the encoding perspective, the GEKs of a convolutional network code (CNC) are shown to be uniquely determined by its LEK matrix $K(z)$ if $K_0$, the constant coefficient matrix of $K(z)$, is nilpotent. This will simplify the CNC design because a nilpotent $K_0$ suffices to guarantee a unique set of GEKs. Besides, the relation between coding topology and $K(z)$ is also discussed. From the decoding perspective, the main theme is to justify that the first $L+1$ terms of the GEK matrix $F(z)$ at a sink $r$ suffice to check whether the code is decodable at $r$ with delay $L$ and to start decoding if so. The concomitant decoding scheme avoids dealing with $F(z)$, which may contain infinite terms, as a whole and hence reduces the complexity of decodability check. It potentially makes CNCs applicable to wireless networks.
\end{itwabstract}

\begin{itwpaper}

\itwsection{Introduction}

Network coding was formally introduced by \cite{NIF2000}. Later, linear network coding was proved to be able to achieve the optimal data transmission rate in an acyclic multicast network \cite{LWC2003}, and an algebraic approach to linear network coding was presented in
\cite{AAATNC2003}. Since then, a rich literature on linear network coding has emerged, and a wide variety of applications have been developed.

Over a network with cycles, the propagation and encoding of sequential data symbols naturally convolve together; the propagation delay thus becomes an essential factor in network coding. To cope with cyclic data transmission, convolutional network coding was introduced in \cite{LWC2003,NWT2005,ENCFCN2005,OCNC2006}. It is a form of linear network coding which deals with the pipeline of messages as a whole rather than individually. Under the assumption of unit-delay edge transmission and local encoding kernels chosen from a finite field, an algebraic framework for a CNC was formulated in \cite{AAATNC2003}. This framework is generalized in \cite{OCNC2006} for those CNCs in which along every cycle there is at least one delay, and the LEKs are chosen among rational power series over the symbol field. The ring of rational power series was justified therein to be the proper algebraic structure for a CNC. More recently, the framework was mathematically extended in \cite{LS2009} for such CNCs in which LEKs determine a unique set of GEKs. In each framework, the main theorem guarantees the existence of an optimal CNC under the respective assumptions such that the GEK matrix at every receiver has full rank. In order to efficiently construct such an optimal CNC over a cyclic network, a polynomial-time algorithm was proposed in \cite{EF2004}. On the other hand, a method was introduced in \cite{LS2009} to adapt any acyclic algorithm for the construction of an optimal CNC on a cyclic network.

Previous studies on convolutional network coding examines every data unit of time-multiplexed data symbols as a whole, and manipulates the LEK matrix over an algebraic structure of data units. In this manner, a CNC is nothing but a linear network code over this algebraic structure. Thus, the classical field-based algebraic framework in \cite{AAATNC2003} for acyclic networks could be applied for a CNC with some additional assumptions and modifications. On the other hand, this facility obscures the implementation aspects of CNCs.

On the encoding side, a CNC is generally deployed in every encoding node by its LEKs. For acyclic networks, GEKs can be uniquely determined by LEKs \cite{NWT2005}, and the practical feasibility is also assured. However, over a cyclic network, the GEKs may not be uniquely deduced from the LEKs (see examples 3.2 and 3.3 in \cite{NWT2005}). Sometimes, the code may not be feasible even if the GEKs can be uniquely deduced because the deduced GEKs do not satisfy the rational form. In this paper, we firstly start the study of a practically feasible CNC from a new approach, that is, the power series representation of a matrix.

We adopt the ring $\mathbb{F}[(z)]$ of rational power series over the symbol field $\mathbb{F}$ as the ensemble of data units for a CNC, where $z$ is the unit time delay. Just as every rational power series over $\mathbb{F}$ can be written as $\sum_{t \ge 0}k_tz^t,k_t \in \mathbb{F}$, we shall represent a matrix $K(z)$ over $\mathbb{F}[(z)]$ as a rational power series $\sum_{t \ge 0}K_tz^t$, where $K_t$ is the matrix over $\mathbb{F}$. For example,
\begin{eqnarray*}
K(z)&=&\left( \begin{array}{*{2}{c}}1 & 1+z^2 \\ 0 & 1+z
\end{array}\right)\\&=&\left( \begin{array}{*{2}{c}}1 & 1\\ 0 & 1
\end{array}\right)+\left( \begin{array}{*{2}{c}}0 & 0 \\ 0 & 1
\end{array}\right)z+\left( \begin{array}{*{2}{c}}0 & 1 \\ 0 & 0
\end{array}\right)z^2
\end{eqnarray*}
With such a novel representation, we can characterize the field-based conditions for a CNC
\begin{itemize}
\item to determine a unique set of GEKs, which is a prerequisite for data propagation; and
\item to be practically feasible.
\end{itemize}
The related work on the encoding side of CNCs is summarized in Fig.\ref{fig:relations} and will be explored in Section \Rmnum{3}, after some related fundamentals on CNCs are reviewed in Sec. \Rmnum{2}. In detail, we first consider the implementation of a CNC and find that a CNC is practically feasible if and only if the operations at intermediate nodes are finite and causal, which is also referred in \cite{OCNC2006}. In order to characterize the property of finite and causal operations, we give the definition of expandability of a matrix. Further more, logical inconsistency exists if there is no partial encoding order along the cycle. So we define an Encoding Topology (\emph{ET}) w.r.t. $K_0$ to illustrate the encoding order. As a result, it is shown that the GEKs can be uniquely formulated from LEKs when the \emph{ET} w.r.t. $K_0$ is acyclic, in which case $K_0$ is nilpotent and results in expandable GEKs. This implies whether a CNC is practically feasible is only determined by $K_0$, which simplifies the CNC encoding design. Finally, we present some equivalent conditions in the same section.
\begin{figure}[h]
  \centering
  \includegraphics[width=7cm]{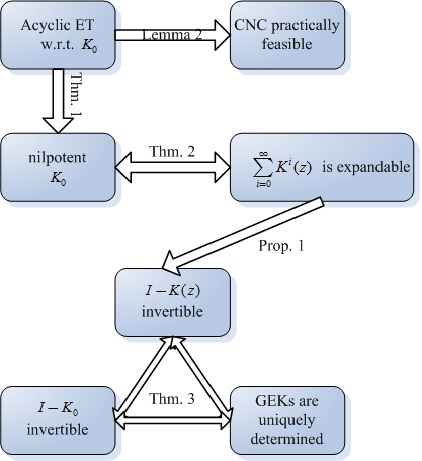}\\
  \caption{the organization and relations among the conditions to uniquely determine GEKs based on the LEK matrix $K(z)=\sum_{i=0}^{\infty}K_iz^i$.}
  \label{fig:relations}
\end{figure}

On the decoding side, in previous work, a CNC is decodable at a receiver if and only if the GEK matrix over $\mathbb{F}[(z)]$ at this receiver has full rank. Moreover, in order to decode with delay $L$, the current method (See \cite{NWT2005} for example) is to compute a matrix $D(z)$ over $\mathbb{F}[(z)]$ such that $F(z)D(z)=z^LI$, where $F(z)$ is the GEK matrix at this receiver and $D(z)$ is called the decoding matrix. However, due to the existence of cycles in the network, GEKs may involve infinite terms. In the decentralized deployment of a CNC, it is infeasible for a receiver to judge whether the code is decodable after collecting all information on its GEK matrix. Even if the GEK matrix is known by the receiver as in the centralized deployment, the computational complexity to check full rank of $F(z)$ will be very high.

In section \Rmnum{4}, based on matrix power series representation of the GEK matrix $F(z)$ for a receiver and motivated by \cite{MS1968}, we give another definition of decodability at a receiver with delay $L$, which captures the feature of sequential transmission of data symbols. Then we provide several sufficient and necessary conditions for the code to be decodable at this receiver with delay $L$, which only involve coefficient matrices $F_0, F_1, \cdots, F_L$ over $\mathbb{F}$. These conditions are field-based and hence greatly reduce the computational complexity in code design. The relations among different conditions for decodability of a CNC are summarized in Fig.\ref{fig:decoding2}.

\begin{figure}[h]
  \includegraphics[width=9.5cm]{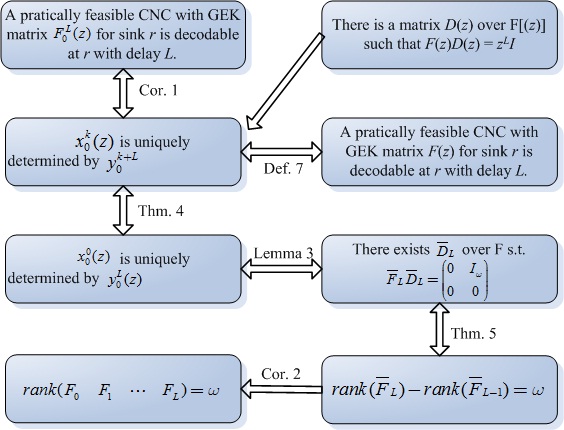}\\
  \caption{the relation among decoding conditions of a CNC. Here $x_0^L(z)$ is the sequence of source symbol vectors generated from time unit 0 to $L$ and $y_0^L(z)$ the sequence of received symbol vectors at receiver $r$ over the same period. Based on matrix power series representation, $F_0^L(z)$ is the first $L+1$ terms of $F(z)$ and $\overline F_L$ is defined by (\ref{Sequ}).}
  \label{fig:decoding2}
\end{figure}

\itwsection{Definitions and related works}
\subsection{Notation}
A communication network is modeled as a finite directed graph
$\mathcal {G}=(\mathcal{V},\mathcal{E})$ with possible cycles.
A directed edge represents a noiseless communication channel
transmitting a data symbol per unit time. Assume there is a unique source
node, denoted by $s$, in the network. The symbol alphabet is a
finite field denoted by $\mathbb{F}$. In each unit time $t$, the source $s$ generates a message, which consists of a fixed number
$\omega$ of symbols and is presented by an $\omega$-dim row vector $x_t$
over $\mathbb{F}$. For every node $v$, denote the set of incoming
channels by $In(v)$, and the set of its outgoing channels by $Out(v)$. For technical convenience, we assume that $In(s)$ consists of $\omega$ imaginary channels. An ordered pair
$\left(d,e\right)$ of channels is called an adjacent pair when there
exists a node $v$ with $d \in In(v)$ and $e \in Out(v)$.

Via a CNC, a stream of messages propagates from $s$
through a time-invariant convolutional encoder at every node, and it is represented by an $\omega$-dim row vector of $x(z)=\sum_{t\ge0}x_tz^t$, where $x_t \in \mathbb{F}^{\omega}$ and $z$ is the time variable. Every entry in $x(z)$ belongs to the principal ideal domain (PID) $\mathbb{F}[[z]]$ of power series. Such power series in $\mathbb{F}[[z]]$, which can be written in the form of $p(z)/(1+zq(z))$, where $p(z)$ and $q(z)$ are polynomials, are called rational power series. Denote the PID of all rational power series by $\mathbb{F}[(z)]$. A square matrix $A(z)$ over $\mathbb{F}[(z)]$ is invertible iff there exists a matrix $C(z)$ over $\mathbb{F}[(z)]$ such that $A(z)C(z)=I$. The determinant of $A(z)$ can be indicated by $\mbox{det}(A(z))=a_0+a_1z+a_2z^2+\cdots \in \mathbb{F}[(z)]$. A sufficient condition for $A(z)$ to be invertible over $\mathbb{F}[(z)]$ is the nonzero value of $a_0$. Denote the adjoint matrix of $A(z)$ by $A^*(z)$. Then, $A(z)A^*(z)=\mbox{det}(A(z))I$. If $a_0$ is nonzero, we have $A(z)^{-1}=A^*(z)/\mbox{det}(A(z))$ and it can be expressed as a
positive power series of $z$. Otherwise, it is not clear for us whether
$A(z)$ is invertible. Assume that $\mbox{det}(A(z))=z^t(a_t+a_{t+1}z+\cdots), \mbox{and }
a_t\neq0,t>0$, the invertibility of $A(z)$ depends on whether the entries of $A^*(z)$ have the common factor $z^t$ or not.

We next define the matrix power series representation of a matrix over $\mathbb{F}[(z)]$.

\newtheorem{definition}{Definition}
\begin{definition}
For an $m \times n$ matrix $A(z)$ over $\mathbb{F}[(z)]$, its matrix power series representation is $\sum_{t=0}^\infty A_tz^t$, where $A_t$ is an $m \times n$ matrix over $\mathbb{F}$ in which the $(i,j)^{th}$ entry is equal to the coefficient of term $z^t$ in the $(i,j)^{th}$ entry in $A(z)$. Two matrix power series are equal if and only if the coefficients of $z^t$ are equal for any $t$.
\end{definition}

\begin{definition}\label{expandable}
Let $A(z)$ be a square matrix over $\mathbb{F}[(z)]$, and $B(z)$ a function of $A(z)$ in the form of $\sum_{t=0}^{\infty} C_t A^t(z)$, where $C_t, t \geq 0$, are square matrices over $\mathbb{F}$. We say that $B(z) = \sum_{t \geq 0} B_t z^t$
is {\em expandable} if for all $t$, the matrix coefficient $B_t$ can
be written as the sum of finite terms in the form of $C_i \prod_jA_j$.
\end{definition}

\vspace{5pt}
Definition \ref{expandable} is motivated by the characterization of practical transmission operations. That is, all the operations at each intermediate node and sink are finite and causal for each time slot $t$. As will be justified in Theorem.\ref{nilpotent F}, if $B(z) = \sum_{t \geq 0} A(z)^t$, then $B(z)$ is expandable if and only if $A_0$ is nilpotent.

\subsection{Fundamentals on CNC}

In this section, we review the local and global encoding kernel descriptions of a CNC.

\begin{definition}
(Local description) An $\omega$-dim $\mathbb{F}$-CNC on a communication network with possible cycles consists of an element $k_{d,e}(z)\in \mathbb{F}[(z)]$, called the local encoding kernel (LEK), for every adjacent pair $\left(d,e\right)$.
\end{definition}

\begin{definition}\label{Def3}
(Global description) A set of global encoding kernels (GEKs) for an $\omega$-dim $\mathbb{F}$-CNC with LEKs $k_{d,e}(z)$ is an assignment of an $\omega$-dimensional column vector $f_e(z)$ over $\mathbb{F}[(z)]$ for every channel $e$ such that:
\begin{enumerate}
\item[1)] $f_e(z)=\sum_{d\in In(v)}k_{d,e}(z)f_d(z)$ when $e\in Out(v)$.
\item[2)] The vectors $f_e(z),e \in In(s)$, form the natural basis of the free module $\mathbb{F}[(z)]^{\omega}$.
\end{enumerate}
\end{definition}

Over an acyclic network, it is equivalent to define a CNC on GEKs and LEKs. However, it is not the case in a cyclic network \cite{NWT2005}. Sometimes LEKs determine multiple sets of GEKs whereas it is also possible for different LEKs to yield a same set of GEKs (See Fig.5 and Fig.7 in Sec. \Rmnum{3} for example). Normality of a CNC was introduced in \cite{LS2009}.

\begin{definition}
A CNC is said to be normal iff LEKs determine a unique set of GEKs.
\end{definition}

The following lemma justifies normality of a CNC as a prerequisite for data transmission via the code.
\newtheorem{lemma}{Lemma}
\begin{lemma}
For a normal CNC, the symbol $y_{e,t}$ transmitted over each channel $e \in Out(v)$ at time $t$ is $\sum_{d \in In(v)} {\left(
{\sum_{\tau = 0}^t  {k_{d,e,\tau } y_{d,t - \tau } } }
\right)}$.
\end{lemma}

\begin{proof}
Let source $s$ generate a message $x(z) = \sum\limits_{t \ge 0} {x_t
z^t }$, which is an $\omega$-dimensional row vector over $\mathbb{F}[(z)]$. Through a CNC, each channel $e$ carries the power series $x(z)\cdot
f_e(z)$ of data symbols. That is,
\begin{eqnarray*}
  y_e (z) &=& x(z) \cdot f_e (z) \\
   &=& x(z) \cdot \sum\limits_{d \in
In(r)} {k_{d,e}(z)f_d (z)} \\
   &=& \sum\limits_{d \in In(r)}
{k_{d,e} (z)\left( {x(z) \cdot f_d (z)} \right)} \\
   &=& \sum\limits_{d
\in In(r)} {k_{d,e}(z)y_d (z)}
\end{eqnarray*}

In matrix power series representation, $y_e (z) = \sum_{t \geq 0} {y_{e,t} z^t }$ and $k_{d,e}(z) = \sum_{t \geq 0} {k_{d,e,t} z^t }$. Then \[
y_{e,t}  = \sum_{0 \le \tau  \le t} {x_\tau  f_{e,t - \tau }
}\]
and \[ y_{e,t}  = \sum_{d \in In(r)} {\left(
{\sum_{\tau = 0}^t  {k_{d,e,\tau } y_{e,t - \tau } } }
\right)}.
\]

\end{proof}

Let $n$ be the number of channels in the
network. For a normal CNC, denote by $K(z)$ the $n\times
n$ matrix $[k_{d,e}(z)]_{d,e \in \mathcal {E}}$, and $F(z)$ the
$\omega \times n$ matrix $[f_e(z)]_{e\in \mathcal{E}}$. Let $H_s$ represent the $\omega \times n$ matrix $[I_{\omega}\; \bold{0}]$. Then we have the classical equation from \cite{AAATNC2003}
\[
\begin{array}{l}
 F(z) = H_s  + F(z) \cdot K(z)
 \end{array}
\]
and we have
\begin{equation}\label{GM}F(z)(I_n-K(z))=H_s\end{equation}

According to (\ref{GM}), if $I_n-K(z)$ is invertible, then the code is normal. A sufficient condition for the invertibility of $I_n-K(z)$ is that $I+K(z)+K^2(z)+\cdots$ is expandable, i.e.,
\begin{eqnarray*}
  (I_n-K(z))(\sum_tB_tz^t) &=&(I-K(z))( \sum_{t=0}^{\infty}K^t(z)) \\
    &=& ( \sum_{t=0}^{\infty}K^t(z))- \sum_{t=1}^{\infty}K^t(z)\\
    &=& I+\sum_{t=1}^{\infty}K^t(z)-\sum_{t=1}^{\infty}K^t(z) \\
    &=& I
\end{eqnarray*}
where $\sum_tB_tz^t=\sum_{t=0}^{\infty}K^t(z)$. That is, $I_n-K(z)$ has an inverse $B(z)=\sum_tB_tz^t$. So we have the following proposition.

\newtheorem{proposition}{Proposition}
\begin{proposition}
$I_n-K(z)$ is invertible if $I+K(z)+K^2(z)+\cdots$ is expandable.
\end{proposition}

However, $I_n-K(z)$ is still possible to be invertible when $I+K(z)+K^2(z)+\cdots$ is not expandable. An example is shown in
Fig. \ref{fig:nonilpotent} in Section \Rmnum{3}.

\subsection{Practical feasibility of a CNC and Encoding topology}

Normality of a CNC is not a sufficient condition for practical implementation of the code because we cannot wait for the whole GEKs to start decoding at the receiver, especially in the case of infinite GEKs. In \cite{OCNC2006,LS2009}, the causality of a CNC is further defined to guarantee physical implementation. However, in decentralized practical application, receivers can only get the GEKs term by term. According to this term by term transformation feature, below we formally justify that causality is a necessary and sufficient condition for a code to be practically feasible in terms of {\em encoding topology (ET)}.

\begin{definition}\label{DefPH}
A normal $\mathbb{F}$-CNC is said to be practically feasible iff each node $v$ can calculate the symbols to be transmitted over its outgoing channels at time slot $t$ by its previously received symbols and LEKs $k_{d,e,\tau}$, where $d\in In(v)$, $e\in Out(v)$, and $\tau \le t$.
\end{definition}

Definition \ref{DefPH} requires that every operation at intermediate nodes is causal and finite. Similar to linear network codes, we would like to find a partial order to define the causal encoding operation. For every CNC on a network, we define {\em ET} with respect to the LEK matrix $K(z)$ to be a directed graph in which $\mathcal{E}$ is the node set and there is a channel from $d$ to $e$ iff $k_{d,e}(z) \not=0$ and take $k_{d,e}(z)$ as the coding weight over the link.

\vspace{5pt}

\noindent\textbf{Example 1.} Assume that we multicast two messages from source $S$ to both the nodes $X$ and $Y$ with given LEKs depicted in Fig. \ref{fig:acyclic}. Easily, we have the \emph{ET} w.r.t. the LEK matrix $K(z)$ shown in Fig. \ref{fig:nilpotent} and the \emph{ET} w.r.t. $K_0$ in Fig. \ref{fig:acyclic1}.

\begin{figure}[h]
    \centering
    \includegraphics[width=8cm,height=6cm]{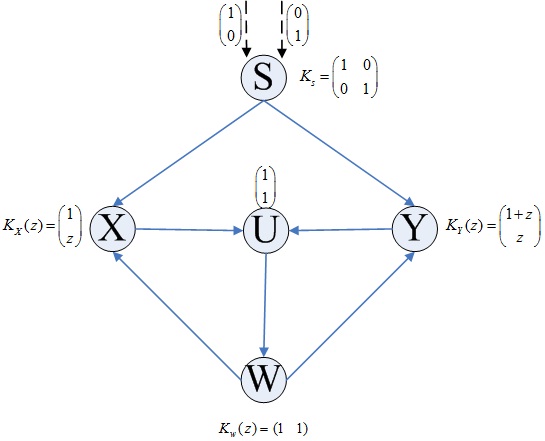}\\
    \caption{A convolutional network code with $K_0$ nilpotent}
    \label{fig:nilpotent}
\end{figure}

\begin{figure}[h]
    \centering
        \includegraphics[width=5.5cm,
        height=4.5cm]{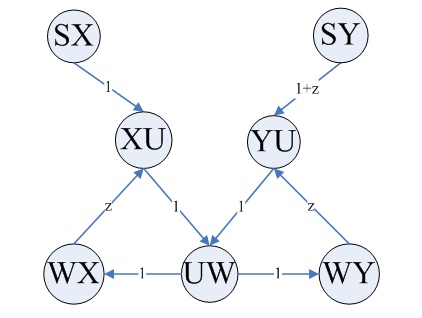}\\
    \caption{The \emph {ET} w.r.t. $K(z)$ in Fig. \ref{fig:nilpotent} is acyclic.}
    \label{fig:acyclic}
\end{figure}

\begin{figure}[h]
    \centering
        \includegraphics[width=5.5cm,
        height=4.5cm]{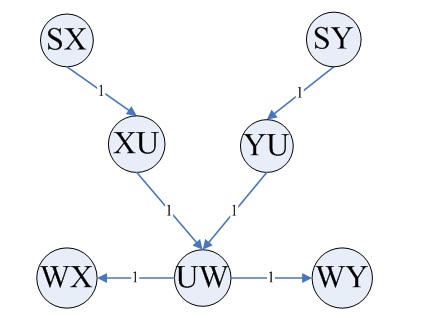}\\
    \caption{The \emph {ET} w.r.t. $K_0$ in Fig. \ref{fig:nilpotent} is acyclic.}
    \label{fig:acyclic1}
\end{figure}

\vspace{5pt}

To perform a CNC on a network we need an {\it encoding order}
$\prec_i$ on the node set for all time slots. We call a set of encoding
orders time invariant if the encoding orders at all time slots
are the same. A CNC is practically feasible if it is practically
feasible w.r.t. a set of encoding orders. \emph{ET} implies the coding relations among edges, and we can designate a partial order when the \emph{ET} is acyclic. Further, we find an acyclic {\em ET} w.r.t. $K_0$ can determine a reasonable order. So we have the following lemma.

\begin{lemma}
A CNC is practically feasible if and only if
the {\em ET} w.r.t. $K_0$ is acyclic.
\end{lemma}
\begin{proof}
Let $k_{d,e}(z)=\sum_{t=0}^{\infty}k_{d,e,t}z^t$ and
$y_e(z)=\sum_{t=0}^{\infty}y_{e,t} z^t$ be LEK for adjacent
pair $(d,e)$ and the transmitted symbols over an outgoing channel $e$ of a node $v$, respectively. Then
form $e \in Out(v)$ by \[
y_e (z) = \sum\limits_{t = 0}^\infty  {y_{e,t} z^t }  =
\sum\limits_{d \in In(v)} {k_{d,e} (z)y_d (z)}.
\]
Therefore, the transmitted symbol of $e$ at time slot $t$ is
\begin{eqnarray}\label{Cconsistent}
  y_{e,t} &=& \sum_{d \in In(v)} \sum_{\tau=0}^t
k_{d,e,t-\tau}y_{d,\tau} \nonumber \\
   &=& \sum_{d \in In(v): d \prec e} \sum_{\tau=0}^t
k_{d,e,t-\tau}y_{d,\tau}  \nonumber \\
   & &+\sum_{d \in In(v): e \prec d} \sum_{\tau=0}^{t-1}
k_{d,e,t-\tau}y_{d,\tau}  \nonumber \\
   & & +\sum_{d \in In(v): e \prec d}k_{d,e,0}y_{d,t}.
\end{eqnarray}
The data symbol $y_{e,t}$ can only be computed after the arrival of all $y_{d, \tau}$, $\tau \leq t$ and $d \in In(v)$ at $v$. At time slot $t$, notice all $y_{d, \tau}$, $\tau \leq t$ have been received by $v$ except those $y_{d,t}$ with $e \prec d$. We can calculate $y_{e,t}$ according to (\ref{Cconsistent}) only if $K_{d,e,0}=0$. That is, a CNC is practically feasible w.r.t. $\{\prec_i\}$ only if $k_{d,e,0}=0$ for all adjacent pairs $(d,e)$ with $e \prec d$.

Hence, according to the partial order $\{\prec_i\}$ defined on the acyclic {\em ET}, all $y_{d,\tau }, d \prec e$ have been calculated for all $\tau  \le t$ because all incoming edges of $v$ are previous to channel $e$. Meanwhile, all $y_{d, \tau}, e \prec d$ are also obtained from previous time slots $\tau <t$. That is, we can compute the output over channel $e$ in terms of the first two items, whose elements are all known at time slot $t$.
\end{proof}

\vspace{5pt}

As for an acyclic network, the {\itshape ET} w.r.t. any LEK matrix $K$ is acyclic. Moreover, $K$ is a strictly
upper triangular matrix, which is also a {\itshape{nilpotent matrix}} which is an $n\times n$ square matrix $K$
such that $K^m=0$ for some positive integer matrix power $m$ \cite{SOTPM1962}. Meanwhile, it is known that the entry in $K^m$ represents the coding gain between any pair of channels through length $m$ path. Hence there must exist an integer $m \leq n$ such that $K^m=0$ because the longest path is bounded in acyclic networks. However, over a cyclic network,
the LEK matrix $K(z)$ is no longer a strictly upper triangular matrix, we will discuss the properties of LEK in this case in the following section.

\itwsection{The conditions to uniquely determine the GEK matrix $F(z)$}

So far, most work in convolutional network coding was developed in the case of a unit time delay network or under the assumption that there
is at least one delay along every cycle. In either case, the \emph {ET} w.r.t.
$K_0$ of a given CNC is acyclic. We will show that $K_0$ is nilpotent, which is also a necessary and sufficient condition to expand
$I+K(z)+K^2(z)+\cdots$. Then the GEKs can be uniquely determined by $K(z)$ of the CNC based on Proposition 1.

\newtheorem{theorem}{Theorem}
\begin{theorem}\label{Thm:nocycle2nilpotent}
$K_0$ is nilpotent if there is no cycle in the \emph {ET} w.r.t. $K_0$.
\end{theorem}
\begin{proof}
We prove by induction. Denote by $K_{ij}^m$ the entry of $K_0^m $ in the $i$-th row $j$-th column, which indicates the coding factor between $i$ and $j$ along the paths of length $m$. The case for $m=1$ is trivial, where $K_0$ represents the one hop transmission matrix. Assume that the statement holds for $m-1$.
By matrix multiplication, we can obtain
\[K_{i,j}^m=\sum_lK_{i,l}^{m-1}K_{l,j}.\] It is easy to see that at least
one of $K_{i,l}^{m-1}$ and $K_{l,j}$ equals to zero for any $l$, because otherwise there must be a path from $i$ to $l$ of length $(m-1)$ and
with $l$ adjacent to $j$, which implies there is a path between $i$ and $j$ of length $m$. Given acyclic \emph {ET} w.r.t. $K_0$, the maximum
paths between arbitrary two nodes are bounded. Hence $K_0$ is nilpotent.
\end{proof}

\vspace{5pt}
The idea of this theorem is referred in \cite{AAATNC2003}.
Inversely, a nilpotent $K_0$ does not imply the acyclic \emph {ET} w.r.t. $K_0$.

\begin{figure}[h]
    \centering
        \includegraphics[width=8cm,
        height=6cm]{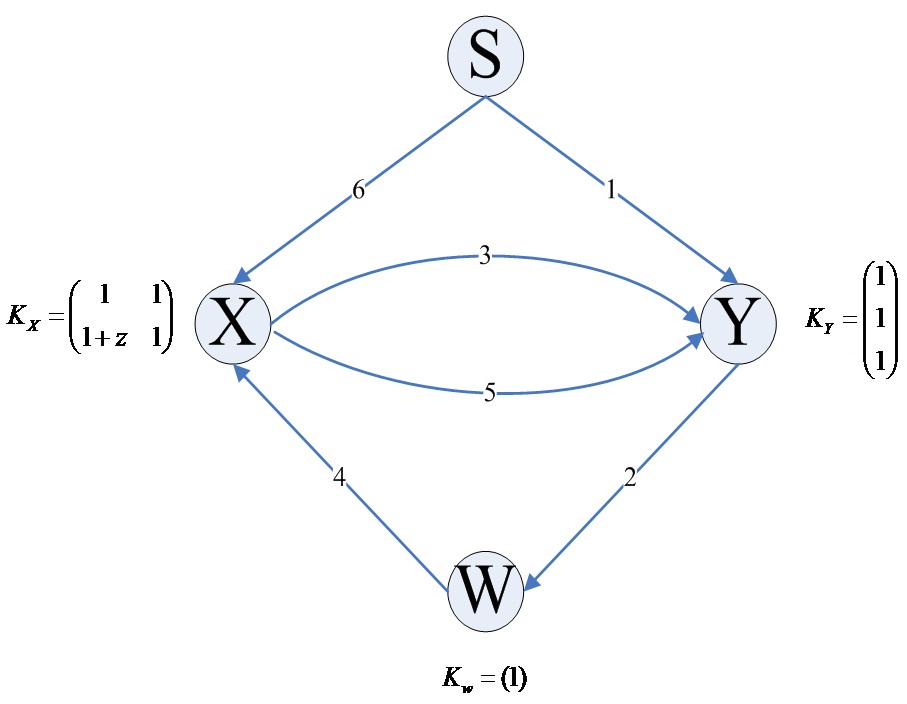}\\
    \caption{a cyclic network with $K_0^m=0$}
    \label{fig:cyclic nilphent}
\end{figure}

\noindent\textbf{Example 2.} Assume that we multicast two messages from source $S$ to both the nodes $X$ and $Y$ with given LEKs depicted by Fig.\ref{fig:cyclic nilphent}. The
local encoding kernel $K_0$ is
\[K_0=\left( \begin{array}{*{6}{c}}0 & 1 & 0 & 0 & 0 & 0\\ 0 & 0 & 0 & 1 & 0 & 0
\\ 0 & 1 & 0 & 0 & 0 & 0\\ 0 & 0 & 1 & 0 & 1 & 0\\ 0 & 1 & 0 & 0 & 0 & 0 \\ 0 & 0 & 1 & 0 & 1 & 0\end{array}\right)\]
and $K_0^4=0$. The corresponding \emph {ET} in Fig.\ref{fig:cyclic} contains
two cycles.

\begin{figure}[h]
    \centering
        \includegraphics[width=5cm,
        height=5cm]{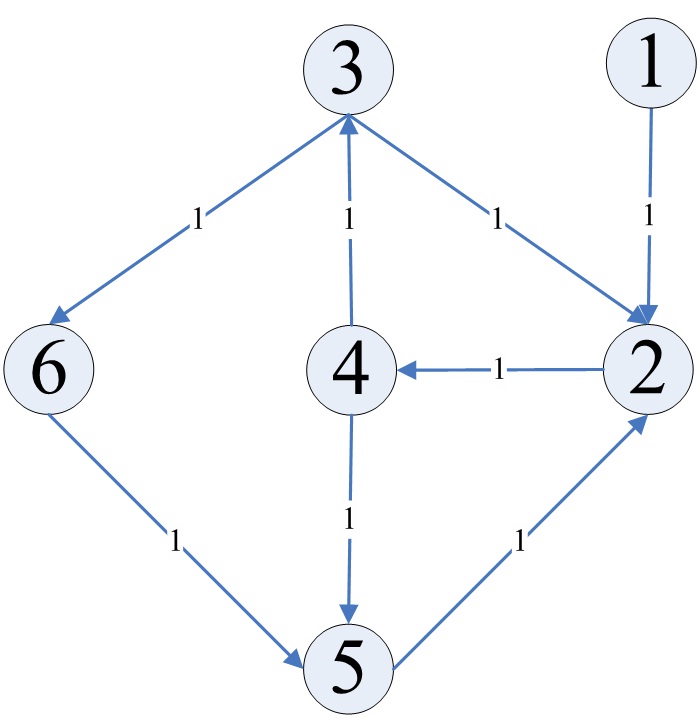}\\
    \caption{the \emph {ET} w.r.t. $K(z)$ in Fig.\ref{fig:cyclic nilphent} is cyclic}
    \label{fig:cyclic}
\end{figure}

In the case of Fig.\ref{fig:cyclic nilphent}, $F(z)$ is practically realizable,
which is as a result of two overlapping cycles. The GEK for channel 2 is derived from channels 1, 3 and 5. The
same message components in channel 3 and 5 are counteracted.
Therefore, the recursion along the cycle is destroyed, and the
logical contradiction can be avoided. This is equivalent to the
acyclic network coding in Fig.\ref{fig:equal encoder}, that is, channel 3 and 5 are actually not the incoming channels of node $Y$ in the calculation of the output over channel 2. However, it
remains unclear whether all the codes with nilpotent $K_0$ can avoid
the logical contradiction. This will be left open for further work.

\begin{figure}
    \centering
        \includegraphics[width=8cm,
        height=6cm]{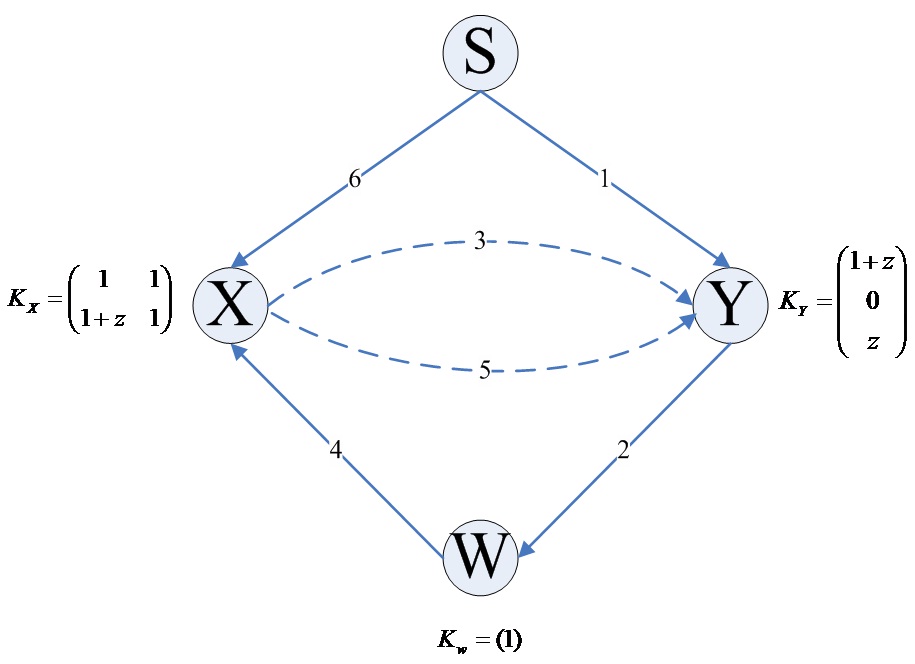}\\
    \caption{an equivalent encoder w.r.t. Fig.\ref{fig:cyclic nilphent}}
    \label{fig:equal encoder}
\end{figure}

We have shown that feasible $F(z)$ yields to expandable
$I+K(z)+K^2(z)+\cdots$. Now we will give an expandable condition of $I+K(z)+K^2(z)+\cdots$.

\begin{theorem}{\label{nilpotent F}}
Let $K_0$ be the constant coefficient matrix of $K(z)$. Then
$I+K(z)+K^2(z)+\cdots$ is expandable if and only if $K_0$ is
nilpotent.
\end{theorem}

\begin{proof}
Assume $A(z)=\sum_{i=0}^{\infty}A_iz^i=I_n+
\sum_{r=1}^{\infty}K^r(z)$. By equating the matrix coefficients, we obtain:
\begin{eqnarray*}
A_0 &=& I_n+K_0+K_0^2+\cdots \\ A_1 &=& K_1+(K_0K_1+K_1K_0) \\&
&+(K_0K_0K_1+K_0K_1K_0+K_1K_0K_0)+\cdots
\\A_2 &=& K_2+(K_0K_2+K_2K_0+K_1K_1) \\ & &+(K_0K_0K_2+K_0K_1K_1+K_0K_2K_0+K_1K_0K_1 \\ & &+K_1K_1K_0+K_2K_0K_0)+\cdots
\\ &\vdots& \\A_i &=&\sum_{j_1+\cdots+j_l=i} \prod {K_0^{m_0}K_{j_1} \cdots
K_0^{m_{(l-1)}}K_{j_l}K_0^{m_l}}\\&\vdots&
\end{eqnarray*}
where $K_{j_{l'}}\neq K_0,1\le j_{l'}\le i$ is the coefficient
matrix of $j_{l'}$ degree in $z$ of $K(z)$. It is obvious that
$A_0$ is finite if and only if $K_0$ is nilpotent. Let $m$ be the
smallest integer such that $K_0^m=0$. Then $m_{l'} \in \left\{
0,1,\cdots,m-1 \right\}$ because the product term containing
$K_0^{m_{l'}}$ will vanish as $m_{l'}\ge m$. To prove
$I+K(z)+K^2(z)+\cdots$ is expandable, we have to prove $A_i$ is the
sum of finite terms for all $i$ as $K_0^m=0$. Now for every $A_i$,
let us consider the number of the terms. Because $j_1+\cdots+j_l=i$
and $1\le j_{l'}\le i$, the number of the solutions is less than
$i^i$, so there are less than $i^i$ choices of these $K_{j_{l'}}$
arrays. Then, for a certain $K_{j_{l'}}$ array, $K_0$ can be divided by
$K_{j_{l'}}$ into $l+1$ parts. So there are at most $m^{l+1}$ choices of
$K_0$ arrays because of $m_{l'}<m$, which is less than $m^{i+1}$.
Therefore, there are less than $i^i\times m^{i+1}$ terms for every
$A_i$.
\end{proof}

\begin{figure}[h]
\centering
        \includegraphics[width=8cm,
        height=6.5cm]{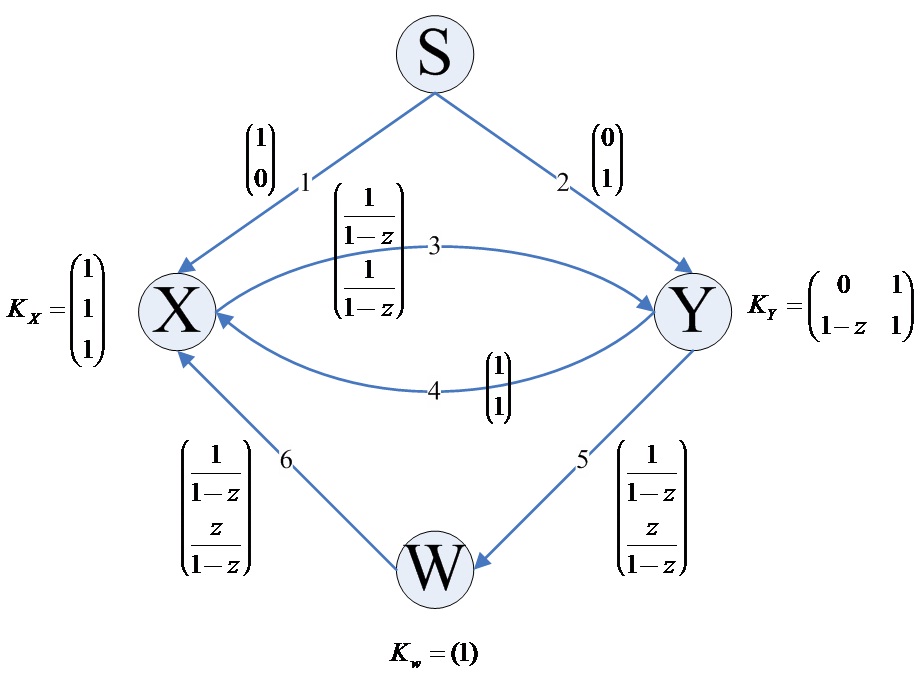}\\
    \hfill\caption{realizable $F\left<z\right>$ with $K_0$ non-nilpotent}
    \label{fig:nonilpotent}
\end{figure}

The expandable $I+K(z)+K^2(z)+\cdots$ implies invertible $I_n-K(z)$. But sometimes, $I_n-K(z)$ may be invertible whereas $K_0$ is non-nilpotent.

\noindent \textbf{Example 3.} As in Fig.\ref{fig:nonilpotent}, the LEK matrix $K(z)$ is \[K(z)=\left(
\begin{array}{*{6}{c}}0 & 0 & 1 & 0 & 0 & 0 \\ 0 & 0 & 0 & 0 & 1 & 0
\\ 0 & 0 & 0 & 1-z & 1 & 0 \\ 0 & 0 & 1 & 0 & 0 & 0 \\ 0 & 0 & 0 & 0
& 0 & 1 \\ 0 & 0 & 1 & 0 & 0 & 0
\end{array}\right),\]
and for $t=0$, $K_0$ is
\[K_0=\left( \begin{array}{*{6}{c}}0 & 0 & 1 & 0 & 0 & 0 \\ 0 & 0 & 0 & 0 & 1 & 0 \\ 0 & 0 & 0 & 1 & 1 & 0 \\ 0 & 0 & 1 & 0 & 0 & 0
\\ 0 & 0 & 0 & 0 & 0 & 1 \\ 0 & 0 & 1 & 0 & 0 & 0 \end{array}\right).\]
It is easy to check that $K_0^6\neq 0$, i.e., $K_0$ is
non-nilpotent.

On the other hand, it is easy to know that $I_n-K_0$ is
invertible, and we have
\begin{eqnarray*}F_0 &=&H_s(I_n-K_0)^{-1}  \\ &=& \left( \begin{array}{*{6}{c}}1 & 0 & 1 & 1 & 1 & 1 \\ 0 & 1 & 1 & 1 & 0 & 0
\end{array}\right)
\end{eqnarray*}

Meanwhile, $F(z)$ can be derived from (\ref{GM}):
\[F(z)=\left[ \begin{array}{*{6}{c}}1 & 0 & 1/(1-z) & 1 & 1/(1-z) & 1/(1-z) \\ 0 & 1 & 1/(1-z) & 1 & z/(1-z) & z/(1-z)
\end{array}\right]\]
We observe that the above $F_0$ and the constant term in the matrix power series representation of $F(z)$ are the same.

\begin{theorem}
Consider an $\omega$-dim $\mathbb{F}$-CNC with LEK matrix $K(z)$. The followings are equivalent.
\begin{enumerate}
\item[1)] $I_n-K_0$ is invertible over $\mathbb{F}$.
\item[2)] $I_n-K(z)$ is invertible over $\mathbb{F}[(z)]$.
\item[3)] $K(z)$ uniquely determines the GEK matrix $F(z)$.
\end{enumerate}
\end{theorem}

\begin{proof}

$1)\Rightarrow 2)$  Since $I_n-K_0$ is invertible over $\mathbb{F}$, then
$\mbox{det}\left(I_n-K_0\right)\neq0$. By using the following property of the determinant repeatedly,

\begin{eqnarray*}
& & \left|\begin{array}{*{4}{c}}a_{11}+b_1 & a_{12}+b_2 & \cdots &
a_{1n}+b_n \\ a_{21} & a_{22} & \cdots & a_{2n}
\\ \vdots & \vdots & \ddots & \vdots \\ a_{n1} & a_{n2} & \cdots & a_{nn} \end{array}\right|  \\
&=& \left|\begin{array}{*{4}{c}}a_{11} & a_{12} & \cdots & a_{1n}\\
a_{21} & a_{22} & \cdots & a_{2n}
\\ \vdots & \vdots & \ddots & \vdots \\ a_{n1} & a_{n2} & \cdots & a_{nn} \end{array}\right|+\left|\begin{array}{*{4}{c}}b_1 & b_2 & \cdots & b_n
\\ a_{21} & a_{22} & \cdots & a_{2n} \\ \vdots & \vdots & \ddots & \vdots \\ a_{n1} & a_{n2} & \cdots & a_{nn} \end{array}\right|
\end{eqnarray*}
we have
\begin{eqnarray*}
\mbox{det}\left(I_n-K(z)\right)&=&\mbox{det}\left(I_n-K_0-K_1z-\cdots\right) \nonumber\\
&=&\mbox{det}\left(I_n-K_0\right)+z^n\mbox{det}\left(K_1\right)
\nonumber \\& &+z^{2n}\mbox{det}\left(K_2\right)+\cdots \nonumber
\\ & &+\mbox{det}\left( \mbox{cross-term matrix} \right) \nonumber \\
&\neq&0
\end{eqnarray*}
which indicates that there is at least one column of the cross-term matrix whose entries
are divisible by $z$, therefore the determinant of the cross-term
matrix is divisible by $z$. Because the constant term of
$\mbox{det}(I-K(z))$ is non-zero, $I-K(z)$ is invertible over
$\mathbb{F}[(z)]$.

$2)\Rightarrow 3)$ Due to equation (\ref{GM}), it is obvious that the GEK matrix $F(z)$ can be uniquely determined by $H_s(I_n-K(z))^{-1}$.

$3)\Rightarrow 1)$ Let $F(z)=\sum_t^{\infty}F_tz^t$ and
$K(z)=\sum_t^{\infty}K_tz^t$. By substituting these into (\ref{GM}), we have
\begin{equation}\label{eq2}
\sum_t(F_t-\sum_{\tau}F_{\tau}K_{t-\tau})z^t=H_s.
\end{equation}
Let the corresponding term coefficients be equal at both sides of the equation, then for $t=0$,
 \begin{equation}\label{F0}
 F_0(I-K_0)=H_s
 \end{equation}
It is sufficient to show that $F(z)$ cannot be uniquely determined by $K(z)$ if $rank(I-K_0)< n$. According to
(\ref{F0}), $F_0$ is uniquely determined by $K(z)$ only if its $i$-th row is the unique solution
of $x_i(I-K_0)=h_i$, i.e., $rank(I-K_0)=n$, where $x_i$ and $h_i$, $i=1,2,\cdots,\omega$, are the $i$-th row of $F_0$ and
$H_s$ respectively. Therefore, $I-K_0$ is invertible over $\mathbb{F}$.

Meanwhile, from (\ref{eq2}), we obtain,
\begin{eqnarray*}
F_1-F_0K_1-F_1K_0&=&0 \nonumber\\
\Rightarrow \qquad F_1&=&F_0K_1(I-K_0)^{-1}  \\
&\vdots & \nonumber \\
F_t(I-K_0)-\sum_{\tau=0}^{t-1}F_{\tau}K_{t-\tau}&=&0 \\
\Rightarrow \qquad
F_t&=&(\sum_{\tau=0}^{t-1}F_{\tau}K_{t-\tau})(I-K_0)^{-1}
\\&\vdots&
\end{eqnarray*}
\end{proof}

\vspace{5pt}

So far, we have discussed the conditions that the LEKs of a CNC can uniquely determine the GEKs. Fig.\ref{fig:relations} summarizes our results mentioned above. The conditions are
characterized in terms of $K_0$. They simplify the convolutional network encoder design and provide a mathematical basis for CNC.

\itwsection{Decoding of CNC}

Consider a practically feasible $\omega$-dim $\mathbb{F}$-CNC with GEK matrix $F(z) = [f_e(z)]_{e \in In(r)}$ at a sink node $r$. Similar to the field-based linear network codes, the $\mathbb{F}$-CNC is decodable at $r$ if and only if
$rank(F(z)) =\omega$. That is, there are $\omega$ incoming channels of $r$ whose GEKs are linearly independent. Moreover, as adopted in \cite{NWT2005,LS2009,LSS2011}, as long as we can find an $|In(r)| \times \omega$ {\em decoding matrix} $D(z)$ {\em over} $\mathbb{F}[(z)]$ such that $F(z)D(z) = z^LI$, the code is {\em decodable with delay L}. Based on $D(z)$, the source symbol vectors can be sequentially decoded via the {\em finite-state linear time-invariant shift registers} \cite{CCIAS1970}. However, the computation of $D(z)$ or even the initial full rank check of $F(z)$ over $\mathbb{F}[(z)]$ might encounter several issues in practice:
\begin{itemize}
\item We need the full knowledge of $F(z)$. However, due to possible existence of cycles in the network, even if the GEKs are polynomials, every entry in $F(z)$ may involve infinite terms. When the CNC is deployed in a randomized manner, the information of $F(z)$ will be carried along with the transmission symbols to the sink node time-slot by time-slot. Thus it will be impossible to know all the terms in every entry of $F(z)$ before checking its decodability, not saying the computation of $D(z)$.
\item Given that $F(z)$ is fully known to the receiver, the computation over $\mathbb{F}[(z)]$ takes high computational complexities due to the possible occurrences of quotient of polynomials.
\item Given that $F(z)$ is of full rank, we can not determine whether it is decodable with delay $L$ untill the calculation of $D(z)$.
\end{itemize}
The existence of a decoding matrix $D(z)$ over $\mathbb{F}[(z)]$ s.t. $D(z)F(z) = z^LI_\omega$ is a sufficient condition for the code's decodability with delay $L$. In order to make the decoding process easier to be handled in practice, especially in the randomized settings, the main goal of this section is to formulate a series of necessary and sufficient conditions on the code's decodability with delay $L$, which will only deal with the first $L+1$ terms $F_0, ..., F_L$ in $F(z) = \sum_{t \geq 0}{F_tz^t}$.

Let $x(z) = \sum_{t \geq 0}x_tz^t$ be the power series of symbol vectors generated at the source and $y(z) = \sum_{t \geq 0}y_tz^t$ the power series of symbol vectors received at sink node $r$. Moreover, denote by $x_0^L(z) = \sum_{t = 0}^L {x_t z^t}$ the sequence of symbol vectors generated from time unit 0 to $L$, and $y_0^L(z) = \sum_{t = 0}^L {y_t z^t }$ the sequence of received symbol vectors over the same period. Specific to the source node $s$ and the sink $r$, a CNC on the network can be regarded as a {\em linear sequential encoder} with encoding kernel $F(z)$. Thus, the definition of invertibility of linear sequential circuits in \cite{MS1968} can be applied to physically define the decodability of a CNC.

\begin{definition}\label{Def5}
A practically feasible $\mathbb{F}$-CNC is decodable at sink $r$ with delay $L$ if and only if for every non-negative integer $k$, the input segment $x_0^k(z)$ is uniquely determined by the
response segment $y_0^{L+k}(z)$.
\end{definition}

Due to the linearity of CNC, it suffices to check whether $x_0^0(z)$ is uniquely determined by $y_0^L(z)$ in order to determine decodability with delay $L$.

\begin{theorem}\label{feasible0}
A practically feasible $\mathbb{F}$-CNC is decodable with delay $L$ if and only if $x_0^0(z)$ is
uniquely determined by the response segment $y_0^L(z)$.
\end{theorem}
\begin{proof}
Necessity follows directly from Definition \ref{Def5}. For sufficiency, suppose that $x_0^0(z) = x_0$ can be uniquely
determined by the sequence of output symbol vectors $x_0F_0, x_0F_1+x_1F_0, x_0F_2+x_1F_1+x_2F_0, \cdots, x_0F_L+\cdots+x_LF_0$ from time unit 0 to $L$. At time unit $L+1$, the output symbol vector is $ x_1F_{L+1}+\cdots+x_{L+2}F_0$. After subtracting the linear gain of $x_0$ from each of output symbol vectors from time unit 1 to $L+1$, we can determine the source symbol vector $x_1$ in the same manner as $x_0$ from the sequence of modified output symbol vectors $x_1F_0, x_1F_1+x_2F_0, \cdots, x_1F_L+\cdots+x_{L+1}F_0$. Following this way, at each time unit $L+k$ after subtracting the respective linear gains of $x_1, \cdots, x_{k-1}$ from the output symbol vectors $y_{L+1}, \cdots, y_{L+k}$, the source symbol vector $x_k$ can be uniquely determined by the modified sequence of output symbol vectors $x_kF_0, x_kF_1+x_{k+1}F_0, \cdots, x_kF_L+\cdots+x_{k+L}F_0$.
\end{proof}

\vspace{5pt}
Since $y_0^L(z) = \sum_{i=0}^Lx_iF_{L-i}z^i$, whether the code is decodable with delay $L$ is only related to the first $L+1$ terms in $F(z)$. Thus,

\newtheorem{corollary}{Corollary}
\begin{corollary}\label{Cor.1}
A practically feasible CNC with GEK matrix $F(z)$ at a sink $r$ is decodable with delay $L$ if and only if a practically feasible CNC with GEK matrix $F_0^L(z)$ at $r$ is decodable with delay $L$.
\end{corollary}

The delay constraint $L$ for the decodability of $F_0^L(z)$ in Corollary \ref{Cor.1} is crucial. For instance, if the GEK matrix $F(z)$ is $\left(\begin{array}{cc} 1+z & 1+z^2 \\ 1 & 1+z\end{array} \right)
$ over $\mathbb{F}_2$, it is not decodable. However, $F_0^1(z) = \left(\begin{array}{cc} 1+z & 1 \\ 1 & 1+z \end{array} \right)$ is decodable with delay 1.

Justified by Theorem \ref{feasible0}, as long as sink $r$ is able to decode the first source symbol vector $x_0$ at a certain time unit $L$, it is able to decode the $k^{th}$ source symbol vector at time unit $L+k$, no matter how the received symbol vector $y_{L+k}$ is
formed. For example, assume the GEK matrix at sink $r$ is
\[F(z)=\left(\begin{array}{cc}
  1 & z \\
  0 & 1+z \\
\end{array}\right)\]
Since $F(z)=F_0+F_1z = \left(\begin{array}{cc}
  1 & 0 \\
  0 & 1 \\
\end{array}\right) + \left(\begin{array}{cc}
  0 & 1 \\
  0 & 1 \\
\end{array}\right)z
,$
it is easy to see that the symbol vector $x_0$ can be decoded at time unit 0 with the field-based decoding matrix $D_0 = \left(\begin{array}{cc}
  1 & 0 \\
  0 & 1
\end{array}\right)$.
Assume that the input is $x(z)=(1,0)+(1, 1)z$. Then the sequence of received symbol vectors is $y(z)= x(z)F(z) = (1,0)+(1,0)z$. At time unit 0, we can recover the first symbol vector $x_0$ by \[y_0D_0 = (1, 0) \left(\begin{array}{cc}
  1 & 0 \\
  0 & 1
\end{array}\right) = (1, 0)\] At time unit 1, we know the linear gain of symbol vector $x_0$ in the received symbol vector $y_1$ is $F_1 = \left(\begin{array}{cc}
  0 & 1 \\
  0 & 1
\end{array}\right)$. The subtraction of $x_0F_1$ from $y_1$ yields $y'_1 = (1, 0)-(0, 1)=(1, 1)$. In the same way as getting $x_0$, the second source symbol vector $x_1$ can be recovered by \[y'_1D_0 = (1, 1)\left(\begin{array}{cc}
  1 & 0 \\
  0 & 1
\end{array}\right) = (1, 1).
\]
This type of soft decision sequential decoding method has been adopted in \cite{EF2004}.

Now, the condition of decodability with delay $L$ has been reduced to check whether $\widetilde{x}_0(z)$ can be recovered from $y_0^L(z)$. Because
\begin{eqnarray}
y(z) &=& x(z)F(z)  \nonumber \\
&=& (\sum_tx_tz^t)(\sum_tF_tz^t)  \nonumber \\
&=& \sum_t(\sum_{i=0}^tx_iF_{t-i})z^t
\end{eqnarray}
we have $y_0^L(z) = \sum_{t=0}^L(\sum_{i=0}^tx_iF_{t-i})z^t$. Alternatively, denote by $\overline{x}_L$ and $\overline{y}_L$, respectively, the row vectors $(x_0, \cdots, x_L)$ and $(y_0, \cdots, y_L)$ of the data symbols generated at source $s$ and received at $r$ from time unit 0 to $L$. Then,
\begin{eqnarray}\label{Decodingeqa}
\overline{y}_L = \overline{x}_L \overline{F}_L,
\end{eqnarray}
where
\begin{equation}\label{Sequ}
    \overline{F}_L = \left(
   \begin{array}{cccc}
     F_0 & F_1 & \cdots & F_L \\
     0 & F_0 & \cdots & F_{L-1} \\
     \vdots & \vdots & \ddots & \vdots \\
     0 & 0 & \cdots & F_0
   \end{array}
 \right)
\end{equation}
Based on $\overline{F}_L$, which is a matrix over $\mathbb{F}$, some necessary and sufficient conditions for recovering $x_0^0(z)$ from $y_0^L(z)$ can be derived.

\begin{lemma}\label{Lem3}
A practically feasible CNC with GEK matrix $F(z)$ for a sink node $r$ is decodable at $r$ with delay $L$ if and only if there is an $|In(r)| \times \omega$ matrix $D(z)$ over $\mathbb{F}[(z)]$ such that $\overline{F}_L\overline{D}_L = \left( \begin{array}{cc}
   \textbf{0} & I_\omega \\
   \textbf{0} & \textbf{0}  \\
\end{array}\right)$, where $\overline{D}_L$ is defined in the same manner as $\overline{F}_L$ from $F(z)$.
\end{lemma}

\begin{proof}
Since $\overline{y}_L\overline{D}_L=\overline{x}_L\overline{F}_L\overline{D}_L=(0, \cdots, 0, x_0)$, the sufficiency follows. For necessity, if $x_0$ can be recovered from $\overline{y}_L$, then there is an $(L+1)|In(r)| \times \omega$ matrix $D = \left( \begin{array}{cccc}
   D_L  \\
  \vdots   \\
   D_0  \\
\end{array} \right)$
over $\mathbb{F}$ such that $\overline{y}_LD = x_0$. Equivalently,
\[
\left(   \begin{array}{cccc}
     F_0 & F_1 & \cdots & F_L \\
     0 & F_0 & \cdots & F_{L-1} \\
     \vdots & \vdots & \ddots & \vdots \\
     0 & 0 & \cdots & F_0 \\
   \end{array}
 \right)
 \left( \begin{array}{cccc}
   D_L  \\
   D_{L-1} \\
   \vdots   \\
   D_0 \\
\end{array}\right) =
 \left( \begin{array}{cccc}
   I_\omega  \\
   \textbf{0}\\
\end{array}\right),
\]
which implies that for every $0 \leq j \leq L-1$, $\sum_{i\geq 0}^jF_0D_i = \textbf{0}$. Therefore, for any $|In(r)| \times \omega$ matrix $D(z)$ with first $L+1$ matrix terms $D_0, \cdots, D_L$ in its matrix power series representation, $\overline{F}_L\overline{D}_L = \left( \begin{array}{cccc}
   \textbf{0} & I_\omega \\
   \textbf{0} & \textbf{0}  \\
\end{array}\right).$
\end{proof}

For technicity, assume $\overline{F}_{-1}$ is a zero $\omega \times |In(r)|$ matrix.

\begin{theorem}\label{Thm5}
A practically feasible $\omega$-dim $\mathbb{F}$-CNC with the GEK matrix $F(z)$ for a sink $r$ is
decodable at $r$ with delay $L$ if and only if
\begin{equation}\label{delaycheck}
rank(\overline{F}_L)-rank(\overline{F}_{L-1})=\omega.
\end{equation}
\end{theorem}

\begin{proof}
We first show the necessity. According to Lemma \ref{Lem3}, there is an $|In(r)| \times \omega$ matrix $D(z)$ over $\mathbb{F}[(z)]$ such that $\overline{F}_L\overline{D}_L = \left( \begin{array}{cccc}
   \textbf{0} & I_\omega \\
   \textbf{0} & \textbf{0} \\
\end{array}\right)$. Consequently,
\begin{equation}\label{bijection_F_D}
(F_0 \ \cdots \ F_L) D = I_\omega
\end{equation}
and
\begin{equation}\label{Kernel_D}
(\textbf{0} \ \overline{F}_{L-1}) D = \textbf{0}.
\end{equation}
where $D = \left( \begin{array}{c}
   D_L \\
   \vdots  \\
   D_0 \\
\end{array}\right)$. According to the matrix property that $rank(AB) \le min \{rank(A), rank(B)\}$, equation (\ref{bijection_F_D}) implies $rank(F_0 \ \cdots \ F_L) \ge \omega$. On the other hand, since $(F_0 \ \cdots \ F_L)$ only has $\omega$ rows, $rank(F_0 \ \cdots \ F_L) \le \omega$. We obtain $rank(F_0 \ \cdots \ F_L) = \omega$. As a result, there does not exist a non-zero $\omega$-dim row vector $\alpha$ subject to $\alpha \cdot (F_0 \ \cdots \ F_L)D = 0 $. In other words, the intersection of the row space and the kernel of column space of $D$ is only the zero vector. However, the row space of $(\textbf{0} \ \overline{F}_{L-1})$ is a subspace of the kernel of the column space of $D$ because of (\ref{Kernel_D}). Therefore, there is no common nonzero vector between the row space of $(F_0 \  F_1 \  \cdots \ F_L)$ and the row space of $(\textbf{0} \ \overline{F}_{L-1})$. The necessity follows.

For sufficiency, if $rank(\overline{F}_L)-rank(\overline{F}_{L-1})=\omega$, then
\begin{enumerate}
\item
the first $\omega$ row vectors in $\overline{F}_L$ have rank $\omega$; and
\item
there is no common nonzero vector between the vector spaces generated by the first $\omega$ and the last $\omega L$ rows in $\overline{F}_L$.
\end{enumerate}

\noindent Denote by $\textit{Null}(\textbf{0} \ \overline{F}_{L-1})$ the null space of last $\omega L$ rows in $\overline{F}_L$ and by $\left<\cdot\right>$ the row space of a matrix. Define a linear transformation $\phi: \textit{Null}(\textbf{0} \ \overline{F}_{L-1}) \rightarrow \mathbb{F}^\omega$ with $\phi(u) = (F_0 \cdots F_L)u$. Prescribed by condition 2 above, \[\left<(F_0 \ \cdots \ F_L)\right> \bigcap \left<(\textbf{0} \ \overline{F}_{L-1})\right> = \{0\}.\]
Moreover, the kernel of $\phi$ is
\[ \{u \in \mathbb{F}^\omega: (\textbf{0} \ \overline{F}_{L-1})u=\textbf{0} \ and \ (F_0 \cdots F_L)u = 0\}.\]
\noindent Thus,
\begin{eqnarray*}
  & &\dim(\ker(\phi)) \\
  &=& |In(r)|(L+1) - rank(\overline{F}_{L-1}) - rank(F_0 \ \cdots \ F_L) \\
   &=& |In(r)|(L+1) - rank(\overline{F}_{L-1}) - \omega.
\end{eqnarray*}

\noindent On the other hand, \[
\dim(\textit{Null}(\textbf{0} \ \overline{F}_{L-1})) = |In(r)|(L+1) - rank(\overline{F}_{L-1}). \]
Hence, the cardinality of the image of $\phi$ is $\mathbb{F}^{\omega}$ and then $\phi$ is surjective. As a result, there exist $\omega$ $|In(r)|(L+1)$-dim column vectors $u_1, \cdots, u_\omega$ over $\mathbb{F}$ such that
\[ (F_0 \ \cdots \ F_L)(u_1 \ \cdots \ u_\omega) = I_\omega, \ and
\]
\[
(\textbf{0} \ \overline{F}_{L-1})(u_1 \ \cdots \ u_\omega) = \textbf{0}.
\]
\end{proof}

By checking (\ref{delaycheck}) iteratively, we can find the minimal decoding delay at sink $r$ which has been characterized in \cite{GC2010}. If the CNC with GEK matrix $F(z)$ for $r$ is known to be decodable with delay $L$ at $r$, an $|In(r)|(L+1) \times \omega$ field-based decoding matrix $\overline{D}_L$ as prescribed in Lemma \ref{Lem3} can be calculated from $\overline{F}_L$. Based on this matrix, we can decode the source symbol vector $x_k$ at time unit $k+L, k > 0$, as follows. Since
\begin{eqnarray*}
(y_k, \cdots, y_{k+L}) = (x_0, \cdots, x_{k+L})
\left(
   \begin{array}{ccc}
     F_k & \cdots & F_{k+L} \\
     \vdots & \cdots & \vdots \\
     F_0 & \cdots & F_{L} \\
     0 & \ddots & \vdots \\
     0 & \cdots & F_0 \\
   \end{array}
 \right),
\end{eqnarray*}

\begin{eqnarray*}
   & & (x_k, \cdots, x_{k+L})\overline{F}_L \\
   &=& (y_k, \cdots, y_{k+L}) - \overline{x}_{k-1}
\left(
   \begin{array}{ccc}
     F_k & \cdots & F_{k+L} \\
     \vdots & \cdots & \vdots \\
     F_1 & \cdots & F_{1+L} \\
   \end{array}
 \right)
\end{eqnarray*}

\noindent Thus, the source symbol vector $x_k$ can be decoded via
\begin{equation}\label{decodingalg}
 \left[(y_k, \cdots, y_{k+L}) - \overline{x}_{k-1} \left(
   \begin{array}{ccc}
     F_k & \cdots & F_{k+L} \\
     \vdots & \cdots & \vdots \\
     F_1 & \cdots & F_{1+L} \\
   \end{array}
 \right) \right] \left(\begin{array}{c}
  D_L \\
  \vdots \\
  D_0 \\
\end{array} \right)
\end{equation}
The field-based decoding algorithm (\ref{decodingalg}) may adopt different decoding matrices $\overline{D}_L$ subject to $\overline{F}_L\overline{D}_L = \left( \begin{array}{cc}
   \textbf{0} & I_\omega \\
   \textbf{0} & \textbf{0}  \\
\end{array}\right)$. For instance, if the GEK matrix is $F(z)=\left(\begin{array}{cc}
  1 & 1 \\
  0 & z \\
\end{array}\right)$ over $\mathbb{F}_2$ and $\overline{F}_1=\left(\begin{array}{cccc}
  1 & 1 & 0 & 0 \\
  0 & 0 & 0 & 1 \\
  0 & 0 & 1 & 1 \\
  0 & 0 & 0 & 0 \\
\end{array}\right)$, then either $\left(\begin{array}{cccc}
  0 & 1 & 1 & 0 \\
  0 & 1 & 0 & 0 \\
  0 & 0 & 0 & 1 \\
  0 & 0 & 0 & 1 \\
\end{array}\right)$ or $\left(\begin{array}{cccc}
  0 & 1 & 0 & 0 \\
  0 & 1 & 1 & 0 \\
  0 & 0 & 0 & 1 \\
  0 & 0 & 0 & 1 \\
\end{array}\right)$ can be adopted to sequentially decode source symbol vectors $x_0, \ x_1, \ \cdots$ with delay $1$.

Similar to the one proposed in \cite{ENCFCN2005}, the present field-based algorithm is time variant. But the complete information on $F(z)$ is not required here. On the other hand, if the matrix $D(z)$ over $\mathbb{F}[(z)]$ is further computed such that $F(z)D(z) = z^LI_{\omega}$, the source symbol vector $x_k$ can be time invariantly decoded at time unit $L+k$ by, in particular, time invariant linear shift registers.

There is a useful consequence of Theorem \ref{Thm5} for preliminary decodability check with lower computational complexity. This has been adopted in the adaptive random construction algorithm for a CNC in \cite{GCSM2011}.

\begin{corollary}
A sink node $r$ with GEK matrix $F_r(z)=\sum_tF_tz^t$ is
decodable with delay $L$ only if
\begin{equation}\label{check1}
rank\left(
        \begin{array}{cccc}
          F_0 & F_1 & \cdots & F_L \\
        \end{array}
      \right)=\omega
\end{equation}

\end{corollary}

This is a necessary but not sufficient condition to decode with
delay $L$ at sink $r$. When determining whether we can start decoding with delay $L$, we first check (\ref{check1}). From the time unit $L$ that (\ref{check1}) is satisfied, we shall turn back to check (\ref{delaycheck}) instead.

\begin{figure}[h]
  \includegraphics[width=8cm,height=7cm]{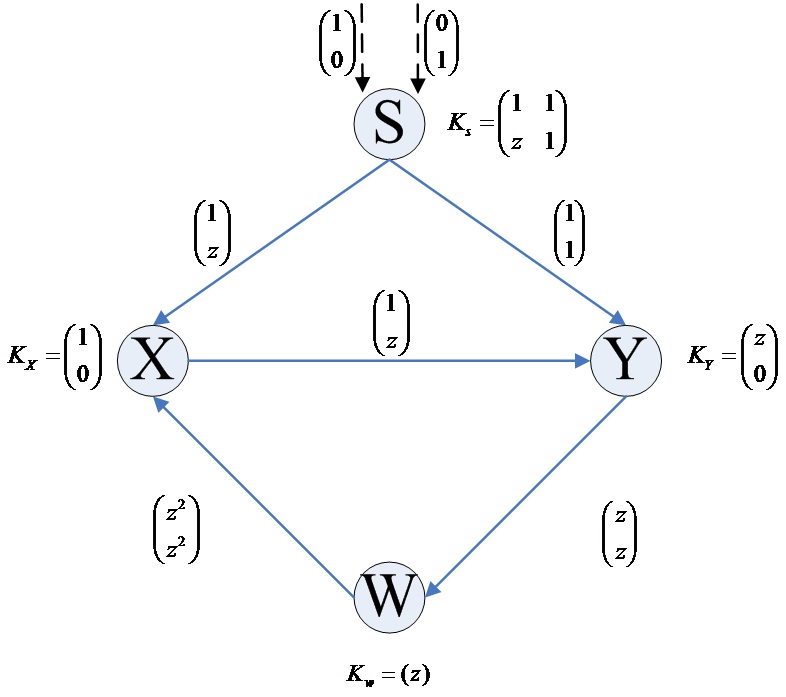}\\
  \caption{A Convolutional Code}\label{NS}
\end{figure}

\noindent \textbf{Example 4.} Fig.\ref{NS} depicts the LEKs and GEKs of a practically feasible CNC on a cyclic network. We shall illustrate the decoding procedure of this code at sink $X$. The GEK matrix for $X$ is
\begin{eqnarray*}
F(z) &=& \left( \begin{array}{*{2}{c}}1 & z^2 \\ z & z^2\end{array}\right) \\
&=&
\left( \begin{array}{*{2}{c}}1 & 0 \\ 0 & 0\end{array}\right)+
\left( \begin{array}{*{2}{c}}0 & 0 \\ 1 & 0\end{array}\right)z+
\left( \begin{array}{*{2}{c}}0 & 1 \\ 0 & 1\end{array}\right)z^2.
\end{eqnarray*}
First, we check at each time unit $t$ the rank of $(F_0 \ \cdots \ F_t)$ until it is equal to $\omega$. An easy check in this example gets $rank(F_0 \ F_1) = \omega$. Then, starting from $t = 1$, we continue to check $rank(\overline{F}_t)-rank(\overline{F}_{t-1})$ until it is equal to $\omega$. At $t = 1$, \[\overline{F}_1=\left( \begin{array}{*{2}{c}}F_0 & F_1 \\ 0 & F_0\end{array}\right)
=\left( \begin{array}{*{4}{c}}1 & 0 & 0 & 0 \\ 0 & 0 & 1 & 0\\0 & 0
& 1 & 0 \\ 0 & 0 & 0 & 0 \end{array}\right).\]
Thus, $rank(\overline{F}_1) = 2$ and $rank(\overline{F}_1) - rank(\overline{F}_0) = 1$. As a result, even if $rank(F_0 F_1)$ is of full rank, we are still not be able to start decoding. At time unit $t=2$, we have
\[\overline{F}_2 = \left( \begin{array}{*{3}{c}}F_0 & F_1 & F_2 \\ 0 & F_0 & F_1\\ 0 & 0 & F_0\end{array}\right)\\
 = \left( \begin{array}{*{6}{c}}1 & 0 & 0 & 0 & 0 & 1 \\ 0 & 0 & 1 & 0 & 0 & 1 \\0 & 0 & 1 & 0 & 0 & 0 \\
 0 & 0 & 0 & 0 & 1 & 0 \\0 & 0 & 0 & 0 & 1 & 0 \\0 & 0 & 0 & 0 & 0 & 0
 \end{array}\right).\]
Hence, $rank(\overline{F}_2)=4$ and $rank(\overline{F}_2) - rank(\overline{F}_1)=2$. This implies that we are able to compute a matrix
\[\overline{D}_2 = \left( \begin{array}{*{3}{c}}D_0 & D_1 & D_2 \\ 0 & D_0 & D_1\\ 0 & 0 & D_0\end{array}\right)\\
 = \left( \begin{array}{*{6}{c}}1 & 0 & 0 & 0 & 1 & 1 \\ 0 & 0 & 1 & 1 & 0 & 0 \\0 & 0 & 0 & 0 & 0 & 0 \\
 0 & 0 & 0 & 1 & 1 & 1 \\0 & 0 & 0 & 0 & 0 & 0 \\0 & 0 & 0 & 0 & 0 & 1
 \end{array}\right).\] such that
\[\overline{F}_2\overline{D}_2 = \left( \begin{array}{cc}
   \textbf{0} & I_\omega \\
   \textbf{0} & \textbf{0}  \\
\end{array}\right)\]
Following this, the source symbol vector $x_0$ can be decoded via $(y_0 \  y_1 \ y_2)\left(\begin{array}{c} D_2 \\ D_1 \\ D_0 \\ \end{array}\right)$ and $x_k$, $k>0$, can be sequentially decoded via (\ref{decodingalg}).

On the other hand, if we plan to start the sequential decoding $x_0, x_1, \cdots$ with delay 2 in a time-invariant manner, we need to find the matrix $D(z)$ such that $F(z)D(z)=z^2I$. Here, it can be deduced that
\[D(z)=\left( \begin{array}{*{2}{c}}z^2/(1-z) & -z^2/(1-z) \\ -z/(1-z) & 1/(1-z)\end{array}\right).\]

\itwsection{Conclusion}

In this paper, we study convolutional network coding by means of matrix power series representation and conclude a few new results. Firstly, for a convolutional network code (CNC) over a single source network with possible cycles, we show that the nilpotent constant coefficient matrix $K_0$ of the LEK matrix $K(z)$ is sufficient to determine whether the code is practically feasible. The encoding topology w.r.t. $K_0$ is also introduced to illustrate the encoding order of the code. Additionally, some equivalent conditions are presented to uniquely determine the GEKs based on $K(z)$.

For decoding of a CNC, we provide a physical definition of decodability at a sink node $r$ with delay $L$, which only involves partial encoding information. Based on this new definition, several necessary and sufficient conditions are established for decodability according to the first $L+1$ terms in the matrix power series representation of the GEK matrix $F(z)$ at $r$. They yield a less computational method for decodability check, and does not require all terms of $F(z)$, which may be infinite because of cycles in the network. As a result, CNC becomes possible to be deployed in a decentralized manner.

There remain several open problems in the study of CNC. One of the major challenges is to design the coding scheme under the more practical scenario that delay conditions vary from time to time. The design of a CNC with minimal decoding delay is another interesting direction for future research.

\itwacknowledgments

This work is partially presented in conference AEW 2010. The first two authors are funded by the National Science Foundation(NSF) under
grant No.60832001. The last author is supported by AoE grant E-02/08 from the University Grants Committee of the Hong Kong SAR, China.

\end{itwpaper}

\begin{itwreferences}
\bibitem{NIF2000}
R. Alshwede, N.Cai, S.-Y. R. Li and R. W. Yeung, ``Network
information flow," \emph{IEEE Transactions on Information Theory},
vol. 46, pp. 1204-1216, Feb. 2000.
\bibitem{LWC2003}
S.-Y. R. Li, R. W. Yeung and N.Cai, ``Linear network coding,"
\emph{IEEE Transactions on Information Theory}, vol. 49, No. 2, pp.
371-381, Feb. 2003.
\bibitem{AAATNC2003}
R. Koetter and M. Medard, ``An algebraic approach to network coding,"
\emph{IEEE/ACM Transactions on networking}, vol. 11, No. 5, Oct.
2003.
\bibitem{MK2002}
M. Medard and R. Koetter, ``Beyond routing: An algebraic approach to
network coding," in \emph{INFOCOM}, vol. 1, pp. 122-130, July 2002.
\bibitem{NWT2005}
R. W. Yeung, S.-Y. R. Li, N. Cai and Z. Zhang, ``Network coding
theory," \emph{Foundation and Trends in Communications and
Information Technology}, vol. 2, ISSN:1567-2190, 2005.
\bibitem{OCNC2006}
S.-Y. R. Li and R. W. Yeung, ``On convolutional network coding,"
\emph{IEEE Transactions on Information Theory}, pp. 1743-1747, Jul.
2006.
\bibitem{CCIAS1970}
G. D. Forney, ``Convolutional codes I: algebraic structure,"
\emph{IEEE Trans. Info. Thy}, vol. 16, pp. 720-738, Nov. 1970.

\bibitem{RFOCNC2008}
S.-Y. R. Li and Siu Ting Ho, ``Ring-theoretic foundation of
convolution network coding," \emph{NetCod2008, CUHK, Hong Kong},
Jan. 2008.

\bibitem{EF2004}
E.Erez and M. Feder, ``Convolutional network codes," \emph{IEEE
International Symposium on Information Theory}, Chicago, June
27-July 2, 2004.
\bibitem{ENCFCN2005}
E. Erez and M. Feder, ``Efficient network codes for cyclic networks,"
\emph{IEEE Trans. Inf. Theory}, vol.56, no.8, pp.3862-3878, Aug., 2010.
\bibitem{LS2009}
S.-Y. R. Li and Q. T. Sun, ``Network Coding Theory via Commutative Algebra," \emph{IEEE Trans. Inf. Theory}, vol.57, no.1, pp.403-415, Jan., 2011.
\bibitem{SOTPM1962}
Ayres and F. Jr, ``Schaum's Outline of Theory and Problems of
Matrices," \emph{New York: Schaum. p. 11}, 1962.
\bibitem{ACBNCCC2004}
C. Fragouli and E. Soljanin, ``A connection between network coding
and convolutional codes," \emph{2004 IEEE Conference on Communications}, pp.
661-666, 2004.
\bibitem{CG2009}
N. Cai and Wangmei Guo, ``The conditions to determine convolutional
network coding on matrix representation," \emph{NetCod2009},
Lausanne, Switzerland, Jun. 2009.
\bibitem{MS1968}
J. L. Massey and M. K. Sain, ``Inverses of linear sequential
circuits," \emph{IEEE Transactions on Computers}, vol. 100, No. 4,
pp. 330-337, Apr. 1968.
\bibitem{GC2010}
Wangmei Guo and N. Cai, ``The minimum decoding delay of convolutional
network coding," \emph{IEICE Trans. on Fundamentals of Electronics, Communications and Computer Sciences}, vol.E93.A, Issue 8, pp. 1518-1523, Aug. 2010.

\bibitem{GCSM2011}
Wangmei Guo, N. Cai, X. Shi and M. Medard, ``Localized Dimension Growth in Random Network Coding: A Convolutional Approach," \emph{ISIT}, St. Petersburg, Russia, Jul 30-Aug. 5, 2011.

\bibitem{LSS2011}
S.-Y.R. Li, Q. T. Sun, Z. Shao, ``Linear network coding: theory and algorithms," \emph{Proc. IEEE}, pp. 372-387, vol.99, no.3, 2011.

\end{itwreferences}

\end{document}